\documentclass[12pt,a4paper,reqno]{amsart}
\usepackage{amsbsy,amsfonts}
\usepackage{amsmath}
\usepackage{amssymb}
\newtheorem{theo}{Theorem}
\newcommand{\pmatrx}[1]{\begin{pmatrix} #1 \end{pmatrix}}

\title{Reciprocal transformations and deformations of integrable hierarchies}

\author{A. Sergyeyev}
\thanks{This research was supported in part
by the Ministry of Education, Youth and Sports of Czech Republic
under grant MSM 4781305904.}
\address{
Silesian University in Opava, Mathematical Institute, Na
Rybn\'\i{}\v{c}ku~1, 746\,01 Opava, Czech Republic}
\email{Artur.Sergyeyev@math.slu.cz}

\keywords{Nonlinear evolution equations,
integrability, Lax representation, $R$-matrix, reciprocal transformations}
\subjclass[2000]{37K10, 37K15}

\begin{document}
\maketitle
\begin{abstract}
\vspace{-10mm}

We present changes of variables that transform new integrable
hierarchies 
found by Szablikowski and B\l aszak using the $R$-matrix deformation
technique [J. Math. Phys. 47 (2006), paper 043505] into known
Harry-Dym-type and mKdV-type hierarchies. 
\end{abstract}

\section*{Introduction}

Recently, Szablikowski and B\l aszak \cite{sb} came up with
a new class of integrable hierarchies which is constructed as follows.
Consider the Lax operator
\begin{equation}\label{lk1}
L = \sum\limits_{i=0}^N u_i D_x^i + D_x^{-1} \circ
u_{-1}
\end{equation}
and the related (formal) spectral problem
\begin{equation}\label{lpsi}
L\psi=\lambda\psi.
\end{equation}
Here $D_{x}=d/d x$ is the total $x$-derivative.

Following \cite{sb} 
consider the following hierarchy of equations compatible with (\ref{lpsi}):
\begin{equation}\label{hier}
\psi_{t_q}=(P_{\geq 1}(L^q)+ \epsilon  \lbrack L^q\rbrack_{0} D_x)(\psi),\qquad q=1,2,\dots
\end{equation}
where for any formal series $K=\sum\limits_{j=-\infty}^k a_i D_x^i$
we have $P_{\geq s}(K)\stackrel{\mathrm{def}}{=}\sum\limits_{j=s}^k
a_i D_x^i$ and $\lbrack K\rbrack_{i}=a_i$, and $\epsilon$ is an
arbitrary constant.

The compatibility conditions for (\ref{lpsi}) and (\ref{hier}) read
\begin{equation}\label{hiera}
L_{t_q}=[P_{\geq 1}(L^q)+ \epsilon  \lbrack L^q\rbrack_{0} D_x,L],\qquad q=1,2,\dots.
\end{equation}
Note that the hierarchy (\ref{hiera}) with $\epsilon=0$
was discovered by Kupershmidt \cite{kup85}.

The goal of the present work is to show that, under a suitable
change of dependent and independent variables $x,t_{i}$, and
$u_{k}$, the ``deformed" hierarchy (\ref{hiera}) can be transformed
into the linear extension of the undeformed ($\epsilon=0$) hierarchy
(\ref{hiera}) in the senses of \cite{kup01}, i.e., into the said
undeformed hierarchy plus a system of linear PDEs on the background
of the latter. We also present a link among the deformed hierarchy
(\ref{hiera}) and the Harry-Dym-type hierarchy (\ref{cc1}). Finally,
we present the dispersionless limit of these results.

\section{Main results}

Let $L$ be as above and 
suppose that we have
\[
[D_{t_i}- P_{\geq 1}(L^i)- \epsilon  \lbrack L^i\rbrack_{0}  D_x,
D_{t_j}- P_{\geq 1}(L^j)- \epsilon  \lbrack L^j\rbrack_{0} D_x]=0,
\quad i,j=0,1,2,\dots.
\]

Let $\chi$ be an arbitrary (smooth) solution of the system
\[
\chi_{t_i}=(P_{\geq 1}(L^i)+ \epsilon  \lbrack L^i\rbrack_{0}
D_x)\chi,\quad i=0,1,2,\dots
\]
such that $\chi_x\neq 0$.

Introduce new set of independent variables $(z,\{\tau_i\})$, where
$z=\chi$ and $\tau_i=t_i$, instead of $(x,\{t_i\})$.

Define new dependent variables $v_i$ by means of the formula
\begin{equation}\label{rt2}
\tilde L=L|_{D_x=\chi_x D_z}= \sum\limits_{i=0}^N u_i (\chi_x D_z)^i
+ D_z^{-1} \circ \chi_x^{-1} u_{-1} \equiv v_N D_z^N+D_z^{-1} \circ
v_{-1}+\sum\limits_{i=0}^{N-1}v_i D_z^i
\end{equation}


We have the following generalization of Lemma 4, i) of
\cite{or}:
\begin{theo}\label{t1}
The transformation $(x,\{t_i\},\{u_j\}) \rightarrow (z,\{\tau_i\},
\{v_j\})$ sends the system
\begin{equation}\label{hier0} L\psi=\lambda\psi,\qquad
\psi_{t_i}=(P_{\geq 1}(L^i)+ \epsilon  \lbrack L^i\rbrack_{0}
D_x)\psi,\quad i=0,1,2,\dots
\end{equation}
into
\begin{equation}\label{hier1}
\tilde L\psi=\lambda\psi,\qquad \psi_{\tau_i}=(P'_{\geq 2}(\tilde
L^i))\psi,\quad i=0,1,2,\dots,
\end{equation}
and hence sends the hierarchy
\begin{equation}\label{cc0}
L_{t_i}=[P_{\geq 1}(L^i)+ \epsilon  \lbrack L^i\rbrack_{0}
D_x,L],\quad i=0,1,2,\dots
\end{equation}
into the Harry-Dym-type hierarchy
\begin{equation}\label{cc1}
\tilde L_{\tau_i}=[P'_{\geq 2}(\tilde L^i),\tilde L],\quad
i=0,1,2,\dots.
\end{equation}
\end{theo}

Here and below 
for any formal series $M=\sum\limits_{j=-\infty}^k b_i D_z^i$ we set
$P'_{\geq s}(M)=\sum\limits_{j=s}^k b_i D_z^i$.

In the undeformed case ($\epsilon=0$) we just recover Lemma 4, i) of
\cite{or}. The proof of Theorem~\ref{t1} amounts to performing the
change of variables $(x,\{t_i\},\{u_j\}) \rightarrow (z,\{\tau_i\},
\{v_j\})$ in (\ref{hier0}) and using the identity (see Lemma 3, i)
of \cite{or})
\[
P'_{\geq 2}(\tilde A)=P_{\geq 1}(A)-(P_{\geq 1}(A)\chi)D_{z}
\]
where $z=\chi$ and $\tilde A=A|_{D_x=\chi_x D_z}$,
valid for any pseudodifferential operator $A$. With this in mind,
we readily see that the transformation in question indeed sends
(\ref{hier0}) into (\ref{hier1}). Finally, as (\ref{cc0}) (resp.\ (\ref{cc1}))
are compatibility conditions for (\ref{hier0}) (resp.\ (\ref{hier1})),
(\ref{hier0}) goes into (\ref{hier1}), as desired.

Note that for $N=1$ a transformation relating the {\em hierarchies}
(\ref{cc0}) and (\ref{cc1}) was found in \cite{gz}.


\begin{theo}\label{rtt}
Consider the {\em reciprocal transformation} from $x$ and $t_i$,
$i=1,2,\dots$ to new independent variables $z$ and $\tau_i$,
$i=1,2,\dots$, where $\tau_i=t_i$, $i=1,2,\dots$, and $z$ is defined
by the formula
\begin{equation}\label{rt1}
dz=(u_N)^{-1/N}dx +\epsilon \sum\limits_{q=1}^{\infty} (u_N)^{-1/N}  \lbrack L^q\rbrack_{0} dt_q,
\end{equation}
and introduce new dependent variables $v_i$ related to $u_i$  by means of the formulas
\begin{equation}\label{rt2}
\tilde L=L|_{D_x=(u_N)^{-1/N}D_z}= \sum\limits_{i=0}^N u_i ((u_N)^{-1/N} D_z)^i + D_z^{-1} \circ
(u_N)^{-1/N} u_{-1} \equiv D_z^N+D_z\circ v_{-1}+\sum\limits_{i=0}^{N-1}v_i D_z^i
\end{equation}
and
\begin{equation}\label{rt3}
v_N=(u_N)^{1/N}.\end{equation}

Then the hierarchy 
(\ref{hiera}) goes into the mKdV-type hierarchy of nonlinear PDEs
for $v_i$,
\begin{equation}\label{hiera1}
\tilde L_{\tau_q}=[P'_{\geq 1}(\tilde L^q),\tilde L],\quad
q=1,2,\dots,
\end{equation}
along with a separate hierarchy for $v_N$:
\begin{equation}\label{hiera1a}
(v_N)_{\tau_q}=-\epsilon (\lbrack \tilde L^q\rbrack_{0})_z,\quad q=1,2,\dots
\end{equation}
\end{theo}


\begin{proof}
First of all notice that \cite{sb} by virtue of (\ref{hiera}) we have
$(u_{N})_{t_{q}}=\epsilon D_{x}(u_N)
[L^q]_{0}- \epsilon N u_N D_{x}([L^q]_{0})$, which is equivalent to
\begin{equation}\label{cl}
(u_{N}^{-1/N})_{t_{q}}=\epsilon D_{x} (u_N [L^q]_{0}).
\end{equation}
From this equality it is immediate that the variable $z$ above is well defined.

Next, pass in (\ref{lpsi}) and (\ref{hier})
from the independent variables $(x,\{t_{q}\})$ to $(z,\{\tau_{q}\})$.
Then (\ref{lpsi}) takes the form
\begin{equation}\label{lpsirt}
\tilde L\psi=\lambda\psi,
\end{equation}
where $\tilde L$ is defined by the first part of (\ref{rt2}), that
is,
\[
\tilde L=L|_{D_x=(u_N)^{-1/N}D_z}= \sum\limits_{i=0}^N u_i ((u_N)^{-1/N} D_z)^i + D_z^{-1} \circ
(u_N)^{-1/N} u_{-1}
\]
while (\ref{hier}) becomes
\begin{equation}\label{hierrt}
\psi_{\tau_q}=P_{\geq 1}(L^q)(\psi),\qquad q=1,2,\dots
\end{equation}
Moreover, if we now introduce
new dependent variables $v_{i}$ using (\ref{rt2})
then we have (see e.g.\ \cite{or})
$P_{\geq 1}(L^q)=P'_{\geq 1}(\tilde L^q)$, so
the compatibility conditions for (\ref{lpsirt}) and (\ref{hierrt})
are precisely (\ref{hiera1}).

On the other hand, if we pass from $(x,\{t_{q}\})$ to $(z,\{\tau_{q}\})$ in
(\ref{cl}) and set $v_{N}=u_{N}^{1/N}$, we obtain
\begin{equation}\label{hiera1a0}
(v_N)_{\tau_q}=-\epsilon D_{z}(\lbrack L^q\rbrack_{0}),\quad q=1,2,\dots,
\end{equation}
and as $P_{\geq 1}(L^q)=P'_{\geq 1}(\tilde L^q)$ implies \cite{or}
$\lbrack L^q\rbrack_{0}=\lbrack \tilde L^q\rbrack_{0}$,
we see that (\ref{hiera1a0}) becomes (\ref{hiera1}), as desired.

Therefore, the hierarchy (\ref{hiera}) goes into (\ref{hiera1})
{\em plus} the equations (\ref{hiera1a}), q.e.d.
\end{proof}

Thus, the hierarchy (\ref{hiera}) does not produce any substantially new integrable
systems: its contents is formed by the linear extensions (\ref{hiera1a}) of the undeformed
mKdV-type hierarchy (\ref{hiera1}).


To show how the above transformation works, set $N=1$ and consider
the so-called extended Broer--Kaup system
 (see Example 4 in \cite{sb}) with $L = u D_x + v + D_x^{-1}\circ w$. From
 $L_{t_i} =
[P_{\geq 1}(L^i)+\epsilon [L^i]_{0}D_x,L]$ with $i=1,2$ we obtain
\begin{eqnarray}
&\pmatrx{u\\ v\\ w}_{t_1} &= \pmatrx{\epsilon u_xv-\epsilon uv_x\nonumber\\ uv_x+\epsilon vv_x
\\ u_xw+uw_x+\epsilon v_xw+\epsilon v w_x}\\
\label{ebk}\\
&u_{t_2} &= \epsilon u_xv^2-2\epsilon uvv_x-2\epsilon u^2w_x-\epsilon u^2 v_{xx}\nonumber\\
&v_{t_2} &= 2uu_xw+2uvv_x+2u^2w_x+uu_xv_x+u^2v_{xx}+\epsilon v^2v_x+2\epsilon uv_xw+\epsilon uv_x^2\nonumber\\
&w_{t_2} &= 2u_xvw+2uv_xw+2uvw_x-u_x^2w-3uu_xw_x-uu_{xx}w-u^2w_{xx}+2\epsilon u_xw^2 \nonumber\\
&\quad &+2\epsilon vv_xw+\epsilon u_xv_xw +\epsilon v^2w_x+4\epsilon uww_x
+\epsilon uv_xw_x+\epsilon uv_{xx}w.\nonumber
\end{eqnarray}
Upon setting $\epsilon =0$ and $u=1$ we
recover \cite{sb} from the second ($t_{2}$) flow the standard Kaup--Broer system.

However, using the transformation from Theorem \ref{rtt} we can reduce the system (\ref{ebk})
to the standard Broer--Kaup system plus a linear equation for $u$. Indeed, let
us pass from the independent variables $x$ and $t_{1}$ to $z$ and $\tau_{1}$ defined by the formulas
$\tau_{1}=t_{1}$, $\tau_{2}=t_{2}$, and
\[
dz=(1/u)dx +\epsilon (v/u)  dt_1+ (\epsilon/u)(2 w u+ u v_{x}+v^{2}) dt_2
\]
(we ignored here the times $t_{i}$ with $i\neq 1,2$).

We have $D_{x}=(1/u) D_{z}$, so $L$ becomes
\[
L=u D_x + v + D_x^{-1}\circ w= D_{z}+ v+ D_{z}^{-1}\circ r,
\]
where $r=w u$.

Upon passing from
the variables $(x,t_{1},t_{2},u,v,w)$ to $(z,\tau_{1},\tau_{2},u,v,r)$
the system (\ref{ebk}) takes the form
\begin{eqnarray}
&\pmatrx{u\\ v\\ r}_{\tau_1} &=
\pmatrx{-\epsilon v_{z}\nonumber\\ v_z\nonumber\\
r_{z}}\\
\label{ebk1}\\
&u_{\tau_2} &= - \epsilon (2 r+ v_{z}+v^{2})_{z}\nonumber\\
&v_{\tau_2} &= 2uu_xw+2uvv_x+2u^2w_x+uu_xv_x+u^2v_{xx}+\epsilon v^2v_x+2\epsilon uv_xw+\epsilon uv_x^2\nonumber\\
&r_{\tau_2} &= 2u_xvw+2uv_xw+2uvw_x-u_x^2w-3uu_xw_x-uu_{xx}w-u^2w_{xx}+2\epsilon u_xw^2 \nonumber\\
&\quad &+2\epsilon vv_xw+\epsilon u_xv_xw +\epsilon v^2w_x+4\epsilon uww_x
+\epsilon uv_xw_x+\epsilon uv_{xx}w.\nonumber
\end{eqnarray}

It is immediate that the last two equations of (\ref{ebk1}), namely those for $v_{\tau_{2}}$ and $r_{\tau_{2}}$,
form the standard (non-extended) Broer--Kaup system for $v$ and $r$, while the equation determining $u_{\tau_{2}}$
is a {\em linear} equation that can be readily solved for given $v$ and $r$.

As a final remark note that the result of Theorem~\ref{rtt} remains valid for the operators
$L$ of more general form than (\ref{lk1}). Namely, Theorem~\ref{rtt} still holds if we replace
$L$ (\ref{lk1}) by
\begin{equation}\label{lk1a}
L = \sum\limits_{i=-\infty}^N u_i D_x^i
\end{equation}
and define new dependent variables $v_{i}$ using the following
modification of the formula (\ref{rt2}):
\begin{equation}\label{rt2a}
\tilde L=L|_{D_x=(u_N)^{-1/N}D_z}= \sum\limits_{i=-\infty}^N u_i ((u_N)^{-1/N} D_z)^i
 \equiv D_z^N+\sum\limits_{i=-\infty}^{N-1}v_i D_z^i.
\end{equation}

\section{Dispersionless case}
Analogs of Theorems~\ref{t1} and \ref{rtt} hold in the
dispersionless case when $D_x$ is replaced by a formal parameter $p$
and the commutator $[,]$ is replaced by the Poisson bracket
\[
\{f ,g\}=\partial f/\partial p D_x(g)-\partial g/\partial p D_x(f),
\]
see e.g.\ the discussion in \cite{sb} and references therein for
further details.

The differential operator $L$ (\ref{lk1}) is now replaced by a
function
\begin{equation}\label{lk1p}
\mathcal{L}= \sum\limits_{i=0}^N u_i p^i + p^{-1} u_{-1},
\end{equation}
and instead of the spectral problem (\ref{lpsi}) we have
\begin{equation}\label{lpsip}
\mathcal{L}|_{p=\psi_x}=\lambda,
\end{equation}
while (\ref{hier}) is replaced by
\begin{equation}\label{hierp}
\psi_{t_q}=(P_{\geq 1}(\mathcal{L}^q)|_{p=\psi_x}+ \epsilon  \lbrack
\mathcal{L}^q\rbrack_{0}|_{p=\psi_x} \psi_x,\qquad q=1,2,\dots,
\end{equation}
where the projection $P_{\geq s}$ is now now defined as follows:
\[
P_{\geq s}\left(\sum\limits_{j=-\infty}^N a_i
p^i\right)=\sum\limits_{j=s}^N a_i p^i.
\]

Finally, the associated hierarchy, the counterpart of (\ref{hiera}),
has the form
\begin{equation}\label{hierap}
\mathcal{L}_{t_q}=\{ P_{\geq 1}(\mathcal{L}^q)+ \epsilon
(\mathcal{L}^q)_{p=0} D_x,\mathcal{L}\},\qquad q=1,2,\dots.
\end{equation}

We have the following results:
\begin{theo}\label{t1p}
The transformation $(x,\{t_i\},\{u_j\}) \rightarrow (z,\{\tau_i\},
\{v_j\})$ sends the system
\begin{equation}\label{hier0p} \mathcal{L}|_{p=\psi_x}=\lambda,\qquad
\psi_{t_i}=(P_{\geq 1}(\mathcal{L}^i)+ \epsilon  \lbrack
\mathcal{L}^i\rbrack_{0}p)_{p=\psi_x},\quad i=0,1,2,\dots
\end{equation}
into
\begin{equation}\label{hier1p}
\tilde{\mathcal{L}}|_{\tilde p=\psi_z}=\lambda,\qquad
\psi_{\tau_i}=(P'_{\geq 2}(\tilde{\mathcal{L}}^i))_{\tilde
p=\psi_z},\quad i=0,1,2,\dots,
\end{equation}
and hence sends the hierarchy
\begin{equation}\label{cc0p}
\mathcal{L}_{t_i}=\{P_{\geq 1}(\mathcal{L}^i)+ \epsilon  \lbrack
\mathcal{L}^i\rbrack_{0} p,\mathcal{L}\},\quad i=0,1,2,\dots
\end{equation}
into the dispersionless Harry-Dym-type hierarchy
\begin{equation}\label{cc1p}
\tilde{\mathcal{L}}_{\tau_i}=\{P'_{\geq 2}(\tilde
{\mathcal{L}}^i),\tilde{\mathcal{L}}\},\quad i=0,1,2,\dots.
\end{equation}
\end{theo}

\begin{theo}\label{rttp}
Consider the {\em reciprocal transformation} from $x$ and $t_i$,
$i=1,2,\dots$ to new independent variables $z$ and $\tau_i$,
$i=1,2,\dots$, where $\tau_i=t_i$, $i=1,2,\dots$, and $z$ is defined
by the formula
\begin{equation}\label{rt1p}
dz=(u_N)^{-1/N}dx +\epsilon \sum\limits_{q=1}^{\infty} (u_N)^{-1/N}
(\mathcal{L}^q)_{p=0} dt_q,
\end{equation}
and introduce new dependent variables $v_i$ related to $u_i$  by
means of the formulas
\begin{equation}\label{rt2p}
\tilde{\mathcal{L}}=\mathcal{L}|_{p=(u_N)^{-1/N}\bar p}=
\sum\limits_{i=0}^N u_i ((u_N)^{-1/N} \bar p)^i + {\bar p}^{-1}
(u_N)^{-1/N} u_{-1} \equiv {\bar p}^N+\bar p
v_{-1}+\sum\limits_{i=0}^{N-1}v_i {\bar p}^i
\end{equation}
and
\begin{equation}\label{rt3p}
v_N=(u_N)^{1/N}.\end{equation}

Then the hierarchy 
(\ref{hierap}) goes into the hierarchy of nonlinear PDEs for $v_i$:
\begin{equation}\label{hiera1p}
\tilde{\mathcal{L}}_{\tau_q}=\{P'_{\geq 1}(\tilde{\mathcal{L}}^q),
\tilde{\mathcal{L}}\},\quad q=1,2,\dots
\end{equation}
and
\begin{equation}\label{hiera1a}
(v_N)_{\tau_q}=-\epsilon (\lbrack
\tilde{\mathcal{L}}^q\rbrack_{0})_z,\quad q=1,2,\dots
\end{equation}
\end{theo}
Note that equations (\ref{hiera1p}) are the compatibility conditions
for the system
\[
\tilde{\mathcal{L}}|_{\tilde p=\psi_z}=\lambda
\]
with
\begin{equation}\label{hier1p}
\psi_{\tau_q}=\left(P_{\geq 1}(\tilde{\mathcal{L}}^q)\right)_{\tilde
p=\psi_z},\qquad q=1,2,\dots.
\end{equation}


\begin{thebibliography}{99}
\bibitem{gz} M. G\"urses, K. Zheltukhin, On a transformation between hierarchies of integrable
equations, Phys. Lett. A  351 (2006),
37--40.


\bibitem{kup85}B.A. Kupershmidt, Mathematics of dispersive water waves, Comm.
Math. Phys. 99 (1985) 51--73.

\bibitem{kup00}B.A. Kupershmidt,
KP or mKP: Noncommutative Mathematics of
Lagrangian, Hamiltonian, and Integrable
Systems, American Mathematical Society (Providence, 2000).

\bibitem{kup01} B.A. Kupershmidt, Dark equations, J. Nonl. Math. Phys. 8 (2001),
363-- 445.

\bibitem{or} W. Oevel and C. Rogers, Gauge Transformations and Reciprocal Links in 2+1 Dimensions,
Rev. Math. Phys. 5 (1993), 299-330.

\bibitem{sb}B. Szablikowski and M. B\l aszak,
J. Math. Phys. 46 (2005), paper 042702 (preprint nlin.SI/0501044 at
arXiv.org).
\end{thebibliography}
\end{document}